\documentclass[manuscript]{acmart}

\usepackage{algorithm}
\usepackage{algorithmic}
\usepackage{dsfont}
\usepackage{subcaption}

\AtBeginDocument{%
  \providecommand\BibTeX{{%
    \normalfont B\kern-0.5em{\scshape i\kern-0.25em b}\kern-0.8em\TeX}}}

\setcopyright{acmcopyright}

\acmSubmissionID{639}

\copyrightyear{2022} 
\acmYear{2022} 
\setcopyright{acmcopyright}\acmConference[FAccT '22]{2022 ACM Conference on Fairness, Accountability, and Transparency}{June 21--24, 2022}{Seoul, Republic of Korea}
\acmBooktitle{2022 ACM Conference on Fairness, Accountability, and Transparency (FAccT '22), June 21--24, 2022, Seoul, Republic of Korea}
\acmPrice{15.00}
\acmDOI{10.1145/3531146.3533188}
\acmISBN{978-1-4503-9352-2/22/06}

\begin{document}

\title{DualCF: Efficient Model Extraction Attack from Counterfactual Explanations}

\author{Yongjie Wang}
\affiliation{%
  \institution{Nanyang Technological University}
  \country{Singapore}}
\email{yongjie002@e.ntu.edu.sg}

\author{Hangwei Qian}
\affiliation{%
  \institution{Nanyang Technological University}
  \country{Singapore}}
\email{QIAN0045@e.ntu.edu.sg}

\author{Chunyan Miao}
\affiliation{%
  \institution{Nanyang Technological University}
  \country{Singapore}}
  \email{ascymiao@ntu.edu.sg}

\renewcommand{\shortauthors}{Yongjie Wang, Hangwei Qian, and Chunyan Miao}

\begin{abstract}

Cloud service providers have launched Machine-Learning-as-a-Service (MLaaS) platforms to allow users to access large-scale cloud-based models via APIs. In addition to prediction outputs, these APIs can also provide other information in a more human-understandable way, such as counterfactual explanations (CF). However, such extra information inevitably causes the cloud models to be more vulnerable to extraction attacks which aim to steal the internal functionality of models in the cloud. Due to the black-box nature of cloud models, however, a vast number of queries are inevitably required by existing attack strategies before the substitute model achieves high fidelity. In this paper, we propose a novel simple yet efficient querying strategy to greatly enhance the querying efficiency to steal a classification model. This is motivated by our observation that current querying strategies suffer from decision boundary shift issue induced by taking far-distant queries and close-to-boundary CFs into substitute model training. We then propose DualCF strategy to circumvent the above issues, which is achieved by taking not only CF but also counterfactual explanation of CF (CCF) as pairs of training samples for the substitute model. Extensive and comprehensive experimental evaluations are conducted on both synthetic and real-world datasets. The experimental results favorably illustrate that DualCF can produce a high-fidelity model with fewer queries efficiently and effectively.      

\end{abstract}



\maketitle

\section{Introduction}

Abundant machine learning models are deployed with automated decision-making ability in various fields such as computer vision~\cite{He_2016_CVPR}, medical diagnosis~\cite{jumper2021highly}, recommender systems~\cite{lu2015recommender}, language translation~\cite{vaswani2017attention}, healthcare~\cite{qian2019novel} and finance~\cite{chen2020deep}. Due to the model/data privacy and computational capacity, the trained models are usually deployed in the cloud via MLaaS platforms, with only public Application Programming Interfaces (APIs) for remote access on a pay-per-query basis. Inevitably, there exists a tension between public accessibility and model confidentiality. On the one hand, the open APIs should be publicly accessible everywhere and anytime. On the other hand, both the datasets and models are intellectual properties of the owners and should be kept private and confidential since (1) model training requires expensive cost on human power, data collection and computation resource; (2) individual's privacy rights should be protected from potential reveals and attacks. Even so, the great commercial value of cloud model steers adversaries to conduct model extraction attack, i.e., to steal the internal functionality of cloud model to construct a substitute model without expensive cost, which facilitates further data tampering to bypass monitoring \cite{xu2016automatically} and stronger attacks, e.g., adversarial attack \cite{papernot2017practical}, model inversion attack \cite{fredrikson2015model} and membership inference attack~\cite{shokri2017membership}.

\begin{figure*}[tb]
    \centering
    \includegraphics[width = 0.8\linewidth]{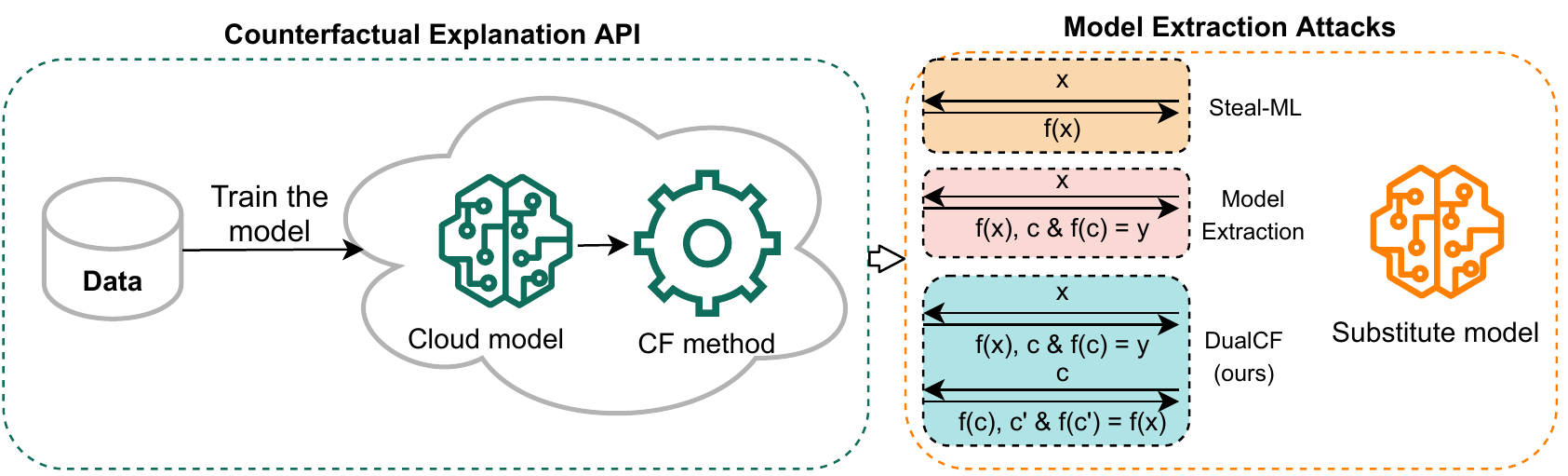}
    \caption{The workflow of model extraction attacks with counterfactual explanations. Service providers train the cloud model on their private data and deploy the counterfactual explanation service in the cloud. Adversaries aim to construct a substitute model by leveraging the information provided by APIs. Three types of model extraction attacks are illustrated.  \textit{Steal-ML}~\cite{tramer2016stealing} only uses the prediction output of the query $\boldsymbol{x}$. \textit{Model Extraction}~\cite{aivodji2020model} considers both the query $\boldsymbol{x}$ and CF $\boldsymbol{c}$ while our proposed DualCF takes CF  $\boldsymbol{c}$ and CCF $\boldsymbol{c}'$ into training. In Model Extraction Attacks module, ``$\leftarrow$'' denotes uploading a query to the API and ``$\rightarrow$'' denotes receiving outcomes from the API.}
    \label{fig:workflow}
\end{figure*}

Model extraction attack aims to obtain a functionally equivalent or near-equivalent machine learning model, which achieves high agreement (up to 100\%) with the cloud model with as fewer queries as possible. We consider real-life scenarios for model extraction attack where service providers (e.g., banks, health centers) deploy machine learning services in the cloud and allow remote access via open APIs. Users can query the API with an input to obtain the corresponding prediction. We assume that (1) the training details, architectures and parameters of the cloud model are invisible to users; (2) users can obtain the data format, the number of classes from the public profile; and (3) users can collect a set of samples to query multiple times. As the API is often on pay-per-query basis and the service provider monitors the abnormal data traffic, fewer queries are preferred. 

Among all these attacks, we focus on the model extraction attack with counterfactual explanations (CF). Counterfactual explanations answer ``what minimum changes are needed for an input instance to alter the current prediction to a particular different one''~\cite{s2017counterfactual}. Especially, for a given input and a pretrained model, counterfactual explanation methods find explanations that have minimal cost to convert the current prediction to a different particular prediction, usually from an undesirable prediction to a desirable one, subject to specified constraints. Note that our proposed approach also uses counterfactual explanation methods to flip desired predictions to undesired ones. This is uncommon for general propose, but such counterfactual explanation methods are capable of flipping predictions of any input. Counterfactual explanations from a desirable prediction to undesirable one reveal the corresponding actions the subject should avoid to prevent the situation from turning worse. As counterfactual explanations offer suggestions with minimum cost to flip the current prediction, therefore, they have broad applications in healthcare (altering an unhealthy situation to a health one), finance (improving the loan approval rate), school admission (obtaining a school offer), paper review (minor revision for paper acceptance), and custom service recovery (regaining the loyalty of unhappy customers). To better illustrate how counterfactual explanations work, we take the loan application as an example: an applicant seeks a mortgage from a loan-granting bank. The applicant submits his/her related information (including age, education, salary, credit score) to the bank. The bank deploys a machine learning model with the binary classifier and then denies this loan application due to the submitted attributes of low salary and poor credit score. Naturally, the individual seeks to know the reason behind the rejection and further to know the necessary changes required before the loan can be approved. The decision-making system equipped with counterfactual explanations is able to provide constructive suggestions in a human-understandable way such as ``increasing the salary by $500$ and increasing the credit score by $100$'' to the applicant, showing the minimal improvement required before the application can be approved. Hence, different from original cloud models merely making opaque algorithmic decisions, cloud models equipped with CF provide additional information to better reflect the underlying key factors that explain the outcome of cloud models. Both attacks are illustrated in Fig.~\ref{fig:workflow}.

While much research \cite{mothilal2020explaining,ustun2019actionable,dhurandhar2018explanations,karimi2020algorithmic} studies how to improve the explanations to meet user requirements, few works study the security and privacy issues of counterfactual explanations, which are essential for safety-critical applications. Despite the intuition that such additional information would naturally enhance the risk of model leakage, up to now, researchers have not reached a consensus on the security and privacy issues of model extraction attacks with CF. Some works~\cite{s2017counterfactual,hashemi2020permuteattack} claim that CF allows individuals to receive explanations without conveying the internal logic of the algorithmic black box since it only conveys a limited set of dependencies on an input, while other researchers~\cite{sokol2019counterfactual,aivodji2020model} do not agree with the above viewpoint. Recently,~\cite{aivodji2020model} experimentally verifies that adversaries can extract a high-fidelity model with counterfactual explanations leaking the information of decision boundary of cloud model. Nevertheless, due to the black-box property of the cloud model, adversaries need to continuously query the API to collect sufficient information before the substitute model can mimic the cloud model with a high agreement. With this in mind, we seek to propose a simple but efficient querying strategy, named \textit{DualCF}, which greatly reduces the number of required queries, and further lowers the overall cost.

\begin{figure*}
    \centering
    \includegraphics[width = \linewidth]{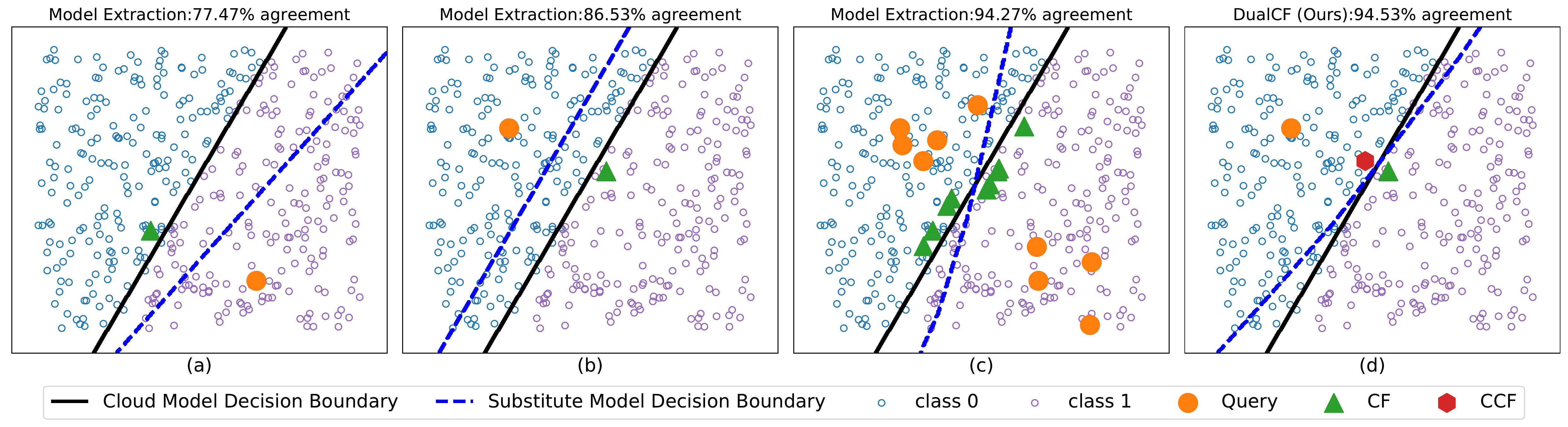}
    \caption{The illustration of decision boundary shift issue. (a) and (b) demonstrate that substitute model decision boundary shifts away from the ground truth due to far-distant queries in existing method \textit{Model Extraction}~\cite{aivodji2020model}. (c) shows that \textit{Model Extraction} has to use more queries to mitigate this issue. Our method~\textit{DualCF} can achieve comparable agreement with only one pair CF and CCF, as shown in (d), which favorably illustrates the efficacy of our method. }
    \label{fig:shift}
\end{figure*}

Our querying strategy~\textit{DualCF} is motivated from the \textit{decision boundary shift issue} in existing \textit{Model Extraction} method~\cite{aivodji2020model}. This issue refers to the decision boundary of the substitute model shifting away from the decision boundary of the cloud model when \textit{Model Extraction} method takes CFs and queries into substitute model training. As it is challenging and almost impossible to obtain the ground-truth decision boundary and training data distribution of the cloud model in advance, the random queries may be far from the decision boundary of cloud model. However, counterfactual explanations are close to the decision boundary regardless of whatever probability threshold of target class is selected for the stop condition. This is because the cloud model tends to assign higher probabilities to all training samples to minimize the training loss. As queries and counterfactual explanations have different predictions, substitute model tends to move towards the middle area between them to separate them confidently. Note that, however, the decision boundary of cloud model is actually close to counterfactual explanations but far from queries. We illustrate this issue in Fig.~\ref{fig:shift}(a) in two-dimensional data space. As queries are randomly selected at each time, the decision boundary of substitute model may shift to different locations, as shown in Fig. ~\ref{fig:shift}(a) and Fig.~\ref{fig:shift}(b). To mitigate this issue,~\textit{Model Extraction} method has to require more queries at a higher attack cost and hence lower efficiency, as shown in Fig.~\ref{fig:shift}(c).

To reduce the number of queries resulting from the decision boundary shift issue, we propose a simple yet efficient method \textit{DualCF}. The gist of our method is to find proper queries with similar distances to decision boundary of cloud model. Note that explicitly selecting the close-to-boundary queries is quite difficult without decision boundary information of cloud model beforehand. To achieve so, our method feeds the current counterfactual explanation (CF) of a query into the open API, and obtains the counterfactual explanation of CF (named CCF for abbreviation). The workflow differences between our DualCF and existing methods (\textit{Steal-ML}~\cite{tramer2016stealing} and \textit{Model Extraction}~\cite{aivodji2020model}) are shown in Fig.~\ref{fig:workflow}. Note that CF and CCF are counterfactual explanations of query and CF respectively, and they tend to locate in the region close to the decision boundary of cloud model but with different predictions. Lastly, we train the substitute model using pairs of CF and CCF. As shown in Fig.~\ref{fig:shift}(d), DualCF with only a pair of CF and CCF achieves comparable agreements as \textit{Model Extraction}. Due to the fact that CF and CCF have the similar distances to the decision boundary of the cloud model and have opposite predictions, the decision boundary shift issue in substitute model training is mitigated. CF and CCF work similarly to support vectors in SVM~\cite{cortes1995support} to help infer the decision boundary. For any query, CF and CCF are always located in the close-to-boundary region, a subset of full data space. Therefore, our method is less sensitive to sampling procedure of queries from full data space.  

Our proposed DualCF greatly reduces the queries theoretically and experimentally for model extraction attack, which consolidates the viewpoint that counterfactual explanations cause the model leakage. Our paper triggers an alarm about the privacy and security of counterfactual explanation service in the cloud and we hope it motivates the corresponding protection strategy in the future. The main contributions of this paper are three-folds:
\begin{itemize}
    \item We dive into the fundamental mechanism of counterfactual explanation generation, and bridge the model extraction attack and counterfactual explanation theoretically. 
    \item We observe that the bottleneck of existing attacks come from decision boundary shift issue, which is caused by training on far-distant queries and CF together. To enhance the efficiency, we propose a simple method DualCF that leverages both the CF and the counterfactual explanation of CF (CCF) into the substitute model training. As CF and CCF have the similar distances to the the decision boundary of cloud model, the boundary shift issue is reduced. 
    \item We conduct extensive experiments on synthetic and real-life datasets. Our study shows that our proposed method can extract the high-fidelity model efficiently and effectively compared with existing methods.
\end{itemize}

\section{Related Work}

\textbf{Counterfactual Explanations.} More complex models are deployed in the cloud for automatic decisions. Due to the black box nature of cloud models, attempts to explain the internal process for a prediction are required to enhance the model trust and rectify negative decisions. Counterfactual explanations~\cite{s2017counterfactual} provide a way to understand the reasons on certain predictions and advice how to make smallest changes to receive a desired prediction. The desired prediction contrasts with the fact (current prediction), which is regarded ``counterfactual''. Counterfactual explanations mainly serve the following purposes: helping users understand why a prediction is made; detecting the model bias for algorithmic fairness~\cite{DBLP:conf/fat/KasirzadehS21}; and providing suggestions to receive a desired result from current decision model~\cite{DBLP:conf/fat/KarimiSV21,ustun2019actionable}. This problem is also studied under other research terminologies like recourse~\cite{ustun2019actionable}, inverse classification~\cite{laugel2017inverse}, and contrastive explanation~\cite{dhurandhar2018explanations}. Following~\cite{s2017counterfactual}, many studies focus on how to model the practical and case-related requirements into mathematical constraints and then solving problems with proper solvers. 

Considering some features are immutable (e.g., race, gender),~\cite{ustun2019actionable} introduces the actionability constraint by freezing the immutable features among improved instance and original query.~\cite{joshi2019towards} requires the probability of counterfactual explanations to follow the data distribution should be large enough to ensure plausible explanations. Diverse explanations are possible to cover multiple choices for each query. Based on this, ~\cite{mothilal2020explaining,russell2019efficient} incorporate diversity constraints on the generated set. As sparser explanations are easier to be adopted for users,~\cite{dhurandhar2018explanations} enforces sparsity property by adding $L_0$ or $L_1$ loss to penalize the changes over many entries of features. In practice, changing a certain feature (e.g., education) may implicitly affect other features (e.g., salary). Hence, it is proper to consider relations~\cite{joshi2019towards,karimi2020algorithmic} or joint effect~\cite{DBLP:conf/fat/PatelSZ21} between feature subsets into modeling. As some features are incomparable,~\cite{wang2021skyline} returns the skyline of non-dominated counterfactual explanations. Once the objective and constraints are determined, it is crucial to develop proper solvers to find satisfied solutions. Typically, different solvers are required for different models and data properties. Besides, varying confidential levels of cloud models require different search process of solvers. For example, integer programming~\cite{ustun2019actionable} or mixed integer programming~\cite{russell2019efficient} solvers are used for linear models with integer or categorical features. If a model is differentiable, gradient descent can be adopted~\cite{s2017counterfactual}. Similarly, model-agnostic method \textit{growing spheres}~\cite{laugel2017inverse} searches the closest counterfactual explanation from the growing sphere around the query. Dijkstra’s algorithm~\cite{poyiadzi2020face} is used for finding a feasible path from the query and the closest counterfactual explanation. Feature tweaking~\cite{tolomei2017interpretable} is designed to retrieve the sub-path leading to target prediction from the current prediction path and select the minimum perturbations for a decision tree model.

\textbf{Model Extraction Attacks.} Model extraction attack intends to train a substitute model that approximates the cloud model well in terms of accuracy and fidelity. The accuracy-oriented methods~\cite{krishna2019thieves,orekondy2019knockoff} aim to create a substitute model that has similar or better performance on a task as the cloud model, while the fidelity-oriented methods \cite{jagielski2020high,juuti2019prada,ijcai2021-336} target to reconstruct a high-fidelity substitute model that approximates the decision boundary of the cloud model as faithfully as possible. In addition to the attacks based on counterfactual explanations, it is also possible to conduct attack with \textit{prediction outputs}, and \textit{gradients}. \textit{Prediction outputs} are the most common outcomes from remote APIs, which consist of the discrete predicted class, predicted class probability and probability vector. Adversaries design various attack algorithms~\cite{papernot2017practical,gong2020model,juuti2019prada,tramer2016stealing,pal2020activethief,yu2020cloudleak} to extract a high-fidelity or high-accuracy model from remote predictions. 
The essential step for these methods is using active learning~\cite{cohn1994improving} strategies to either generate informative synthetic points \cite{papernot2017practical,juuti2019prada,tramer2016stealing,yu2020cloudleak} or select max-coverage points \cite{pal2020activethief}. The ``informative points'' represent the data points close to the decision boundary (i.e., adversarial examples~\cite{goodfellow2014explaining}) of the cloud model, and the ``max-coverage points'' mean the data points should be far distant between each other (i.e., core-set \cite{sener2018active}). \textit{Gradients} help explain the model behavior upon an infinitesimal perturbation~\cite{simonyan2013deep}, a.k.a. the sensitivity in the neighborhood. As such, gradients can be used to explain the feature importance for a prediction of a given input.~\cite{milli2019model} finds that gradient-based feature importance methods can easily expose the cloud model to adversaries. In particular, the gradient of a arbitrary instance is the model weights for a linear model. 

Despite the fact that counterfactual explanations do not disclose the cloud model in its entirety, the security and privacy of CFs have largely been overlooked~\cite{DBLP:conf/fat/KasirzadehS21,DBLP:conf/fat/BarocasSR20}. Some works~\cite{s2017counterfactual,hashemi2020permuteattack} claim that counterfactual explanations cannot expose the internal algorithmic logic except a limited set of dependencies on a single instance. Even though such limited information is trivial for extraction, gathering sufficient information with more queries is prominent to conduct model extraction. 
\cite{sokol2019counterfactual} points out that counterfactual explanations disclose more secrets of cloud model and can enhance the attack efficiency. A recent study~\cite{aivodji2020model} firstly conducts the model extraction attack on counterfactual explanations by viewing these close-to-boundary explanations as additional training instances, but we observe that it suffers from the decision boundary shift issue caused by far-distant queries especially when the query size is small. Therefore, it requires higher attack cost to query more times. Our proposed DualCF mitigates the boundary shift issue by introducing a novel querying strategy, which also greatly reduces the querying cost.


\section{Preliminaries}
\label{sec:statement}

A cloud model $f_{\theta}:\mathcal{X}\in \mathbb{R}^d \to \mathcal{Y}\in \mathbb{R}$ takes an arbitrary input query $\boldsymbol{x} \in \mathcal{X}$ as the input, and predicts the output $y \in \mathcal{Y}$. In this paper, we assume the cloud model is a pretrained neural network model for classification, parametrized by frozen weights $\theta$. The counterfactual explanation method $g: f_{\theta}\times \mathcal{X} \to \mathcal{X}$ generates a minimal perturbed instance $\boldsymbol{c}\in \mathcal{X}$ for the input instance $\boldsymbol{x}$ such that $f(\boldsymbol{c})$ has the desirable prediction. Formally, searching for the counterfactual explanation $\boldsymbol{c}$ can be framed as the following mathematical formulation,
\begin{align}
\label{eq:definition}
    \mathop{\arg\min}_{\boldsymbol{c}}\ &\ d(\boldsymbol{x}, \boldsymbol{c}) \\
    s.t.\  &\ f(\boldsymbol{c}) = y\text{ and } f(\boldsymbol{c}) \neq f(\boldsymbol{x}) \nonumber
\end{align}
where $d(\cdot): \mathcal{X}\times \mathcal{X} \rightarrow \mathbb{R}_{+}$ is a distance (cost) metric measuring the changes between the input $\boldsymbol{x}$ and $\boldsymbol{c}$, and $y$ is the desirable target which is different from the original prediction $f(\boldsymbol{x})$. That is, we seek to find counterfactual explanations that belong to a target class $y$ while still remains proximal to the original instance. In addition to the constraints in Eq~\eqref{eq:definition}, more constraints such as sparsity, feasibility and diversity can be added according to task-specific requirements~\cite{verma2020counterfactual}, which are left for future work.

In this paper, we focus on the high-fidelity extraction attack on counterfactual explanations which aims to extract a functional equivalent or near-equivalent model $h_{\phi}$ that behaves very similarly to the model $f_{\theta}$. It can be formulated as the following mathematical problem: for a set of queries $\mathcal{D}$ and a set of corresponding counterfactual explanations, it finds a substitute model $h_{\phi}$ that performs equivalently on an evaluation set $\mathcal{T}$.
\begin{align}
\label{eq:obj}
    \max & \sum_{\boldsymbol{x}_i \in \mathcal{T}}\mathds{1}_{f_{\theta}(\boldsymbol{x}_i) = h_{\phi}(\boldsymbol{x}_i)}\\
    s.t.\ & h_{\phi}(\boldsymbol{x}) = f_{\theta}(\boldsymbol{x}), \boldsymbol{x} \in \mathcal{D} \\ 
         & \boldsymbol{c} = g(f_{\theta},\boldsymbol{x}), \boldsymbol{x} \in \mathcal{D} \\
         & h_{\phi}(\boldsymbol{c}) = f_{\theta}(\boldsymbol{c})
\end{align}
The existing method~\cite{aivodji2020model} feeds a query $\boldsymbol{x}$ into the API and then obtains the prediction $f(\boldsymbol{x})$ and counterfactual explanation $\boldsymbol{c}$. The pairs of $(\boldsymbol{x}, f(\boldsymbol{x}))$ and $(\boldsymbol{c}, f(\boldsymbol{c}))$ are treated as the training set without discrimination to learn the substitute model $h_{\phi}$. 

\section{Proposed Approach}
\label{sec:metohds}

\subsection{Relation between Model Extraction Attack and Counterfactual Explanations}

In this section, we observe that counterfactual explanations explicitly reveal not only the decision boundary location of cloud model but also important features of the cloud model. Such information is favorable for model extraction attack. 

Here, we illustrate the information leakage through construction of counterfactual explanations. Let $\{x_1, x_2, ..., x_d\}$ denote $d$-dimensional features of an instance $\boldsymbol{x}$, where each feature $x_i$ is associated with a feature importance value $w_i$ w.r.t. cloud model's prediction, e.g., $w_i$ is $i$-th coefficient if $f_{\theta}$ is a linear model. A larger $w_i$ leads to a larger prediction change if we tweak the feature $x_i$ by the same amount. For simplicity, we consider the scenario with $2$ features ($d=2$) and $w_1 > w_2$. In Eq.~\eqref{eq:definition}, note that the prediction change $\Delta f = f(\boldsymbol{c}) - f(\boldsymbol{x})$ from current prediction $f(\boldsymbol{x})$ to desired prediction $f(\boldsymbol{c})$ should be strictly satisfied, and the explanations with minimal cost (measured by some distance metrics) are preferred. We write the objective on 2d example as,
$$\mathop{\arg\min}_{\boldsymbol{c}} ||(c_1 - x_1, c_2 - x_2)||_p, \text{ s.t. } (c_1 - x_1, c_2 - x_2) (w_1, w_2)^T = \Delta f.$$
As $w_1$ is larger in our assumption, smaller change $(c_1 - x_1)$ can satisfy the prediction change in the constraint and then associates with a smaller cost in the objective. Therefore, the ideal counterfactual explanation should change the feature $x_1$ first. From the above analysis, we can conclude that counterfactual explanation methods naturally give high priority to change the features that have larger importance to the desirable prediction. 

Essentially, model extraction attack seeks to infer the decision boundary of cloud model from the parameter space. With the favorable information from counterfactual explanations besides prediction output, the search process of substitute model is accelerated. As counterfactual explanations and queries have different predictions, there must exist a decision boundary to separate them. In addition, the feature with minor change corresponds to the important features in the cloud model. Therefore, we can infer the full or partial ranking of important features, which also help enhance the attack efficiency. The closeness of counterfactual explanations to decision boundary of cloud model should be also unearthed for efficient attacks. 

\subsection{Decision Boundary Shift Issue by Far-distant Queries}
\label{subsec:over_confidence}

The straightforward method takes counterfactual explanations and queries as training samples of the substitute model, as pointed in~\cite{aivodji2020model}. However, we observe that this method suffers from the decision boundary shift issue, i.e., the substitute model's decision boundary shifts away from the ground truth. The issue is more severe especially when query size is smaller. To relieve the issue, existing methods have to adopt more queries, which result in higher querying cost. Here, we illustrate where this issue comes from and how it affects the attack efficiency. 

In the beginning, adversaries know nothing about the decision boundary of cloud model and therefore the queries to upload may be far from the decision boundary of cloud model. However, counterfactual explanations produced by cloud model are usually close to decision boundary, and in the other side of the decision boundary. When we train the substitute model with counterfactual explanations and queries, the decision boundary of the substitute model tends to move to the middle area of the query and CF to separate them confidently. This deviates from the fact that decision boundary of the cloud model is close to counterfactual explanations and may be far from queries. As shown in Fig.~\ref{fig:shift}(a) and (b), the substitute model's decision boundary is distorted by far-distant queries. Due to the sampling variance in queries, the substitute model's decision boundary may move to different regions. This results in a unstable substitute model for an attack method. That is why existing methods need to use more queries to achieve high fidelity and algorithm robustness. 

\begin{figure*}
    \centering
    \includegraphics[width = \linewidth]{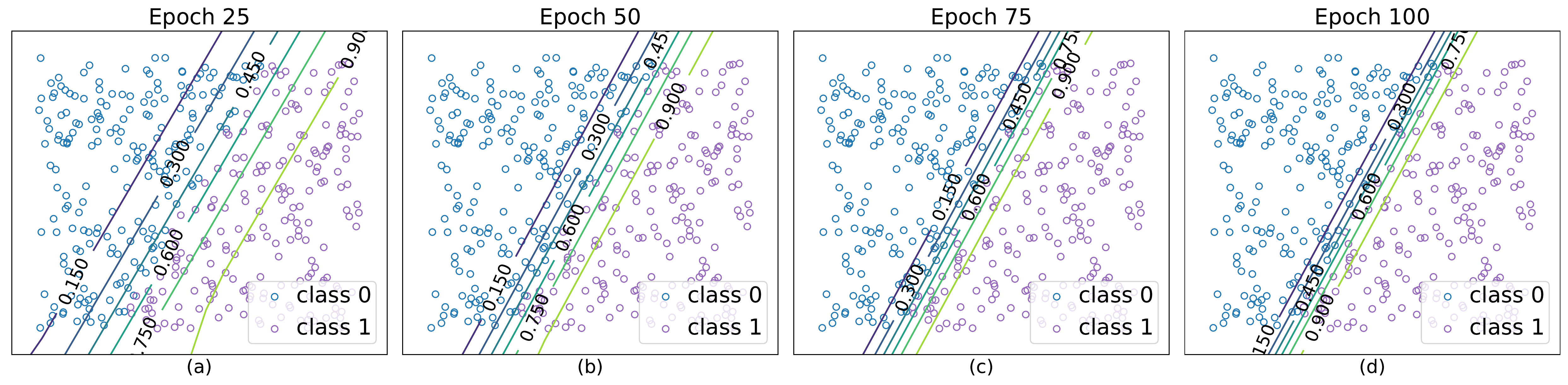}
    \caption{The probability density contour lines during training stage (at epoch 25, 50, 75, 100) of the cloud model trained on a synthetic dataset for binary classification task. The blue and purple dots are sampled points from training set. The values on the contour lines represent the probability calculated by the sigmoid function with range being $[0,1]$. We can see the model tends to assign high probabilities to more samples during training.}
    \label{fig:saturation}
\end{figure*}

For probabilistic models, the constraint $f(c) = y$ in Eq.~\eqref{eq:definition} becomes $p(y|c) > \epsilon$, meaning that the counterfactual explanation method searches the closest explanations above a given probability threshold $\epsilon$ of target class. As shown in Fig.~\ref{fig:saturation}(a), regions closer to decision boundary correspond to probability around $0.5$, wherein the model is uncertain about its predictions. Regions far away from decision boundary, on the contrary, correspond to either smaller or larger values, and the values indicate more confident predictions on class 0 and 1, respectively. In this case, could a higher probability threshold $\epsilon$ push counterfactual explanations away from the decision boundary to reduce the issue? The answer is yes, but the fact is that even if we set a higher probability threshold, counterfactual explanations are still close to the decision boundary. We observe that this issue roots in the over-confident predictions phenomenon in cloud models.  That is, deep neural networks tend to assign high probability to all training instances, resulting in that instances with high probability can be close to the decision boundary, as shown in Fig.~\ref{fig:saturation}(d). We discuss this phenomenon in the following.  

For a general neural network model for classification, the probability of predicted label belonging to class $k \in [1, K]$ is computed as $p(k) = \frac{\exp(z_k)}{\sum_{i=1}^K \exp(z_i)}$ or $p(k) = \sigma(z_k)$, where $z_i$, $z_k$ represent the logits in the last layer of the neural network. The cross-entropy loss for classification is defined as $\ell =-\sum_{k=1}^K q(k)\log(p(k))$, where $q$ is a one-hot vector for ground truth. Minimizing this loss function is equivalent to maximize the likelihood of ground truth, in other words, $z_k >> z_i$ for $i \neq k$ is desired if the ground truth class is $k$. As all training points are undifferentiated during training, the model tends to assign higher probability on actual class entry for each instance in order to minimize the loss function. This brings the over-confident prediction problem: the probability increases quickly and flattens if we gradually move an instance away from the decision boundary. We illustrate this with a binary classification model on a synthetic dataset in Fig.~\ref{fig:saturation}. The figure shows that more and more instances (i.e., instances in the left side of $<10\%$ contour and instances in the right side of $>90\%$ contour) have high probability ($>90\%$ for class 1 or $<10\%$ for class 0) and hence the probability density contour lines become denser and denser during training. Therefore, simply adjusting the probability threshold $\epsilon$ cannot ease the issue since counterfactual explanations are consistently close to the decision boundary due to the over-confident prediction problem.

\subsection{DualCF}

To overcome the decision boundary shift issue resulting from far-distant queries, we propose our algorithm \textit{DualCF} by considering both CF and the counterfactual explanation of CF (CCF) as pairs of training samples for substitute model. As discussed previously, the decision boundary of the substitute model is distorted by taking far-distant queries and close-to-boundary counterfactual explanations into training. Our intuition is that, if training samples with different classes have similar distances to the decision boundary of cloud model, then this issue can be alleviated. As counterfactual explanations locate in the close-to-boundary region, could we search close-to-boundary queries directly? Actually, we do not know the decision boundary of cloud model in the beginning. Without decision boundary information, it is impossible to find the close-to-boundary queries directly. However, we can treat the counterfactual explanation of an arbitrary query as another query to the cloud API and obtain a CCF with the same prediction of original query. Since both CF and CCF are produced by the same counterfactual explanation API, CF and CCF have the similar distances to the decision boundary of the cloud model naturally. In addition, CCF and CF reside in opposite regions of the decision boundary due to their different predictions of classes. As such, the decision boundary shift issue is mitigated by adopting CF and CCF into substitute model training. 

The proposed \textit{DualCF} is listed in Algorithm \ref{alg:algorithm}. Our method first uploads a query $x$ into the counterfactual explanation API and obtain an explanation $\boldsymbol{c}$ (line 5). Secondly, our method treats current explanation $\boldsymbol{c}$ as another query and obtains a CCF $\boldsymbol{c}'$ (line 6). Similarly, we generate the pair of CF and CCF for each of $N$ queries and denote them as $\mathcal{S}$. Finally, we train the substitute model $h_{\phi}$ with $\mathcal{S}$. It needs to be emphasized that $\boldsymbol{c}$ and its corresponding $\boldsymbol{c}'$ should be used in the same batch in order to mitigate the decision boundary shift issue. Hence, our proposed DualCF is a simple yet efficient querying strategy. The major difference between DualCF and recent studies \cite{tramer2016stealing,aivodji2020model} is how to construct the training set for $h_{\phi}$. Steal-ML~\cite{tramer2016stealing} collects $(\boldsymbol{x}, f(\boldsymbol{x}))$, Model Extraction~\cite{aivodji2020model} adopts $(\boldsymbol{c}, f(\boldsymbol{c}))$ and $(\boldsymbol{x}, f(\boldsymbol{x}))$ while our proposed DualCF uses $(\boldsymbol{c}, f(\boldsymbol{c}))$ and $(\boldsymbol{c}', f(\boldsymbol{c}'))$. Note that the initial query $\boldsymbol{x}$ is not used as training set in our DualCF. We also introduce a variant of our proposed method, denoted as DualCFX, which takes $\boldsymbol{x}$ as training data as well, i.e., $(\boldsymbol{c}, f(\boldsymbol{c})), (\boldsymbol{c}', f(\boldsymbol{c}'))$, and $(\boldsymbol{x}, f(\boldsymbol{x}))$. 

\begin{algorithm}[tb]
\caption{DualCF}
\label{alg:algorithm}
\begin{algorithmic}[1]
\STATE \textbf{Input}: Queries $\{\boldsymbol{x}_i\}_{i=1}^N$, cloud model $f_{\theta}$, counterfactual explanation API $g(\cdot)$.\\
\STATE \textbf{Output}: Substitute model $h_{\phi}$.
\STATE Initialize empty training set for substitute model $\mathcal{S} = \{\}$.
\FOR{i$<$ N}
\STATE $\boldsymbol{c} = g(f_{\theta}, \boldsymbol{x})$ 
\STATE $\boldsymbol{c}' = g(f_{\theta}, \boldsymbol{c})$ 
\STATE $\mathcal{S}$ = $\mathcal{S} \bigcup \{ (\boldsymbol{c}, f(\boldsymbol{c})), (\boldsymbol{c}', f(\boldsymbol{c}'))\}$
\ENDFOR
\STATE cur\_iter = 0
\WHILE{cur\_iter $<$ max\_iter}
\STATE Minimize the training loss of $h_{\phi}$ on dataset $\mathcal{S}$.
\STATE cur\_iter++.
\ENDWHILE
\STATE \textbf{return} $h_{\phi}$.
\end{algorithmic}
\end{algorithm}

We visualize the differences between existing attack methods and DualCF in Fig.~\ref{fig:toy_example} on a binary dataset. Two classes (shown as blue and purple dots) are separated by the model $f_{\theta}$. Three orange dots are three sampled queries, red hexogons and green triangles are CFs and CCFs respectively. The queries, corresponding CFs and CCFs are numbered by $0,1,2$. Steal-ML~\cite{tramer2016stealing} with prediction output only achieves lowest agreement. Model Extraction adds counterfactual explanations into training and achieves higher agreement. We also notice the decision boundary are pulled away from the ground truth by far-distant queries in Fig.~\ref{fig:toy_example}(c). With CFs and CCFs as training pairs, our method obtains a substitute model with the highest agreement.  

\begin{figure*}
    \centering
    \includegraphics[width = \linewidth]{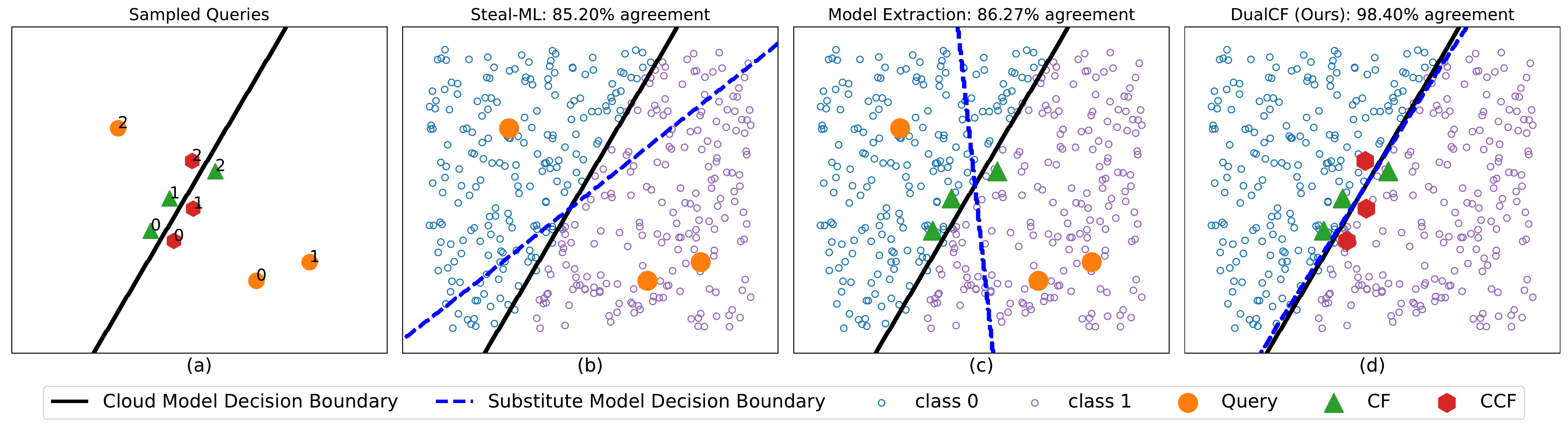}
    \caption{The toy examples on a binary synthetic dataset. We show $3$ queries and corresponding CFs and CCFs (orange dots, green triangles, and red hexagons) in figure (a). Figure (b), (c) and (d) illustrate the substitute model's decision boundary produced by 3 attack methods. Steal-ML uses prediction outputs of queries, Model Extraction applies both queries and CFs, and our proposed DualCF considers CFs and CCFs. Our proposed \textit{DualCF} achieves the highest agreement.}
    \label{fig:toy_example}
\end{figure*}

From the above analysis, we can see the advantages of DualCF. CF $\boldsymbol{c}$ and CCF $\boldsymbol{c}'$ are both close to decision boundary and have different predictions, not only reducing the boundary shift issue but also bringing a tighter space for inferring the decision boundary of cloud model. The CF $\boldsymbol{c}$ and CCF $\boldsymbol{c}$ work as the support vectors in SVM algorithm~\cite{cortes1995support}. Secondly, the number of each class is balanced favoring the learning process. Thirdly, our method is less sensitive to the sampling procedure of the queries because CF and CCF locate in the denser region close to the decision boundary. 

\subsection{DualCF for A Linear Model}

Here, we illustrate how our DualCF extracts a functionally equivalent model from only a pair of CF and CCF for a linear model. Suppose we have a binary linear model $f_{\theta} = \sigma (\langle \theta, x \rangle + b)$ , where the decision boundary is determined by the parameter $\theta$ and $\sigma$ is the sigmoid function. For a query $\boldsymbol{x}$, its counterfactual explanation $\boldsymbol{c}$ can be optimally found by searching along the direction of gradient ($f_{\theta}'(x) = f_{\theta}(x) * (1-f_{\theta}(x)) * \theta$), stopping when it reaches a certain probability $\epsilon$ belonging to the target class. Using the API, we obtain the CF and CCF of a query. 

\begin{lemma}
For a binary linear model $f_{\theta}$, we can extract a substitute model $h_{\phi}$ with $100\%$ agreement, from a pair of CF $\boldsymbol{c}$ and CCF $\boldsymbol{c}'$ given an input $\boldsymbol{x}$.
\end{lemma}

\begin{proof}

As the model is linear, the CF and CCF can be optimally found along or against the gradient direction. We draw a straight line through the input point $\boldsymbol{x}$ which is perpendicular to decision boundary. The intersection point of the straight line and decision boundary is denoted as $\boldsymbol{x}_0$. Then, the CF and CCF can be written as
\begin{align}
  \boldsymbol{c} &= \boldsymbol{x}_0 + a \theta \\
  \boldsymbol{c}' &= \boldsymbol{x}_0 - a \theta
\end{align}
respectively, where $a$ depends on the probability threshold for a counterfactual explanation. As we have $\boldsymbol{c}$ and $\boldsymbol{c}'$ on hand, we can obtain a point $\boldsymbol{x}_0 = \frac{\boldsymbol{c} + \boldsymbol{c}'}{2}$ on the decision boundary. Besides, from any two of $\boldsymbol{x}$, $\boldsymbol{c}$ and $\boldsymbol{c}'$, we can obtain the slope $\theta$. With slope and a point in the line, we can obtain the decision boundary without any training. The lemma holds. 
\end{proof}

Although extending the lemma to a nonlinear model remains unsolved, this linear scenario verifies that our proposed method can replicate the cloud model efficiently. In experiment section, we evaluate proposed DualCF in more general situations where the cloud model is nonlinear and complex.


\section{Experiments}

In this section, we conduct extensive experiments to evaluate the proposed \textbf{DualCF} and its variant \textbf{DualCFX} on both synthetic and real-life datasets to compare with state-of-the-art methods. 

\subsection{Baselines}
We employ the following state-of-the-art methods to verify the performance of our method. 
\begin{itemize}
    \item \textbf{Steal-ML} \cite{tramer2016stealing}. This method first labels the queries with the cloud model $f_{\theta}$, and then trains the model $h_{\phi}$ on the labelled dataset. We randomly select queries for this method. 
    \item \textbf{Steal-ML (CoreSet)} \cite{pal2020activethief}. As random selection may choose redundant or similar samples that do not bring more useful information for extraction, we consider the second baseline that leverages the CoreSet algorithm in \cite{ijcai2021-336} to select the most distant samples to query the cloud model. 
    \item \textbf{Model Extraction} \cite{aivodji2020model}. This method is proposed to extract a substitute model for counterfactual explanation methods. It utilizes the prediction of queries and counterfactual explanations from the cloud model to train the substitute model. 
\end{itemize}

\subsection{Implementation Details}
We first train a Multilayer Perceptron (MLP) model $f_{\theta}$ as the cloud model for each dataset. For the counterfactual explanation method, we use the existing algorithm~\cite{s2017counterfactual} implemented in DiCE \cite{dice_github} (without diverse requirement) for the cloud model. The default objective function $d(\cdot)$ is the $L2$ distance metric on normalized features. We display counterfactual explanations from DiCE on synthetic datasets in Fig.11 in Appendix A.1. After that, we train another MLP model $h_{\phi}$ to approximate cloud model $f_{\theta}$. For simplicity, cloud model and substitute model have the same architecture, but we train both models with different random initializations and different data. In the following experiments, we also study the influence of substitute models with different capacities from cloud model. Adam optimizer is used for minimizing the binary cross-entropy loss of cloud model and substitute model. We stop the substitute model training until it reaches a maximum epoch. For fair comparison, all substitute models for all baselines and proposed methods on each dataset have the same training settings and architectures. We introduce more details of substitute model training in Appendix A.2.

\subsection{Evaluation Metrics}
As we target to build a high-fidelity model $h_{\phi}$ that behaves as similar to $f_{\theta}$ as possible, we firstly define the \textbf{agreement} which measures the prediction difference between the $f_{\theta}$ and $h_{\phi}$ on a evaluation set of size $n$,
\begin{align}
    \text{Agreement} = \frac{1}{n}\sum_{i=1}^{n}\mathds{1}_{f_{\theta}(\boldsymbol{x}_i) = h_{\phi}(\boldsymbol{x}_i)}.
\end{align}
A higher agreement is better for substitute models on the same queries. To show how the query size influences the algorithm behavior, we gradually select increasing number of queries and report the agreement on them. A higher agreement curve is better. To reduce the experiment variance due to random sampling of queries and random initialization in model training, we compute the average agreement over 100 runs for a fixed query size. We also compare the \textbf{standard deviation (std)} of 100 runs to measure the algorithmic stability. The smaller standard deviation means that the attack method is more stable to the sampling procedure and model initialization. If two model extraction methods have the same agreement, then lower standard deviation is preferred. 

\subsection{Datasets}

We consider the following $5$ datasets, which are widely used in research on counterfactual explanations for evaluation. 

\begin{itemize}

\item \textbf{Synthetic Dataset}. We generate two 2-dimensional datasets where all data points are uniformly sampled from a close interval $[0, 6]$. We draw a straight line on the first dataset and a S-curve on the second dataset, to separate data points of each dataset into two parts and assign each part with a label. The first dataset is named \textbf{Syn-Linear} and the second dataset is named \textbf{Syn-Nonlinear} respectively. 

\item \textbf{Give Me Some Credit (GMSC)} \cite{gmsc_dataset}. This dataset is collected for predicting whether someone will experience finance distress in the next two years by his/her financial and demographic information (10 numerical features). The full dataset contains $150,000$ applicants where $139,974$ applicants are labeled as ``good'' and $10,026$ applicants are labeled as ``bad''. We select the first $10,026$ ``good'' records and all ``bad'' records to form the final dataset for balance. We follow the pre-processing procedure in \cite{gmsc_preprocessing} to fill the missing values, remove outliers and delete irrelevant features, etc. 

\item \textbf{Heloc Dataset} \cite{heloc_dataset}. It is used for predicting whether an individual will repay the Heloc account in two years by $23$ numerical features describing the personal information. The target variable ''RiskPerformance'' is binary, 5000 ``good'' records indicating no default , 5459 ``bad'' records indicating the opposite. We keep the top-10 important features based on the analysis \cite{heloc_preprocessing} of the IBM team, which is the champion of FICO Heloc challenge.  We adopt the same pre-processing in \cite{heloc_preprocessing} to remove abnormal values. 

\item \textbf{Boston Housing Dataset} \cite{harrison1978hedonic}. The dataset has $506$ records, where each record has 12 features for predicting the house price in Boston. We follow the pre-processing of the alibi tutorial \cite{boston_preprocessing}, which firstly transforms the continuous labels into binary classes based on whether the house price is above the median or not, and then removes three categorical features. 

\end{itemize}

We split these datasets into $3$ disjointed sets: training, query and evaluation sets at the ratio of $50\%$, $25\%$, $25\%$. We describe them in detail as follows: (1) \textbf{training set} (50\%) is used for training the cloud model $f_{\theta}$; (2) \textbf{query set} (25\%) is used to feed into counterfactual explanation API $g(\cdot)$ for obtaining counterfactual explanations, and we train the substitute model $h_{\phi}$ on them. Note that we merely upload a subset of queries to counterfactual explanation API in our experiments; (3) \textbf{evaluation set} (25\%) is for evaluating how faithful the substitute model $h_{\phi}$ is. As we target to build a high-fidelity model, we use the predicted label instead of the true label on the evaluation set. 
A validation set for tuning the cloud model $f$ is omitted because we assume the cloud model is given and frozen. Therefore, we directly use the architecture and hyper-parameters from public to train the cloud model $f_{\theta}$. For the cloud model, we normalize the features by standard normalization where the mean value and variance are estimated on the training set. For the substitute model, we compute the mean value and variance on the full query set for standard normalization. 


\subsection{Experimental Results}

\begin{figure*}[tb]
    \includegraphics[width = \linewidth]{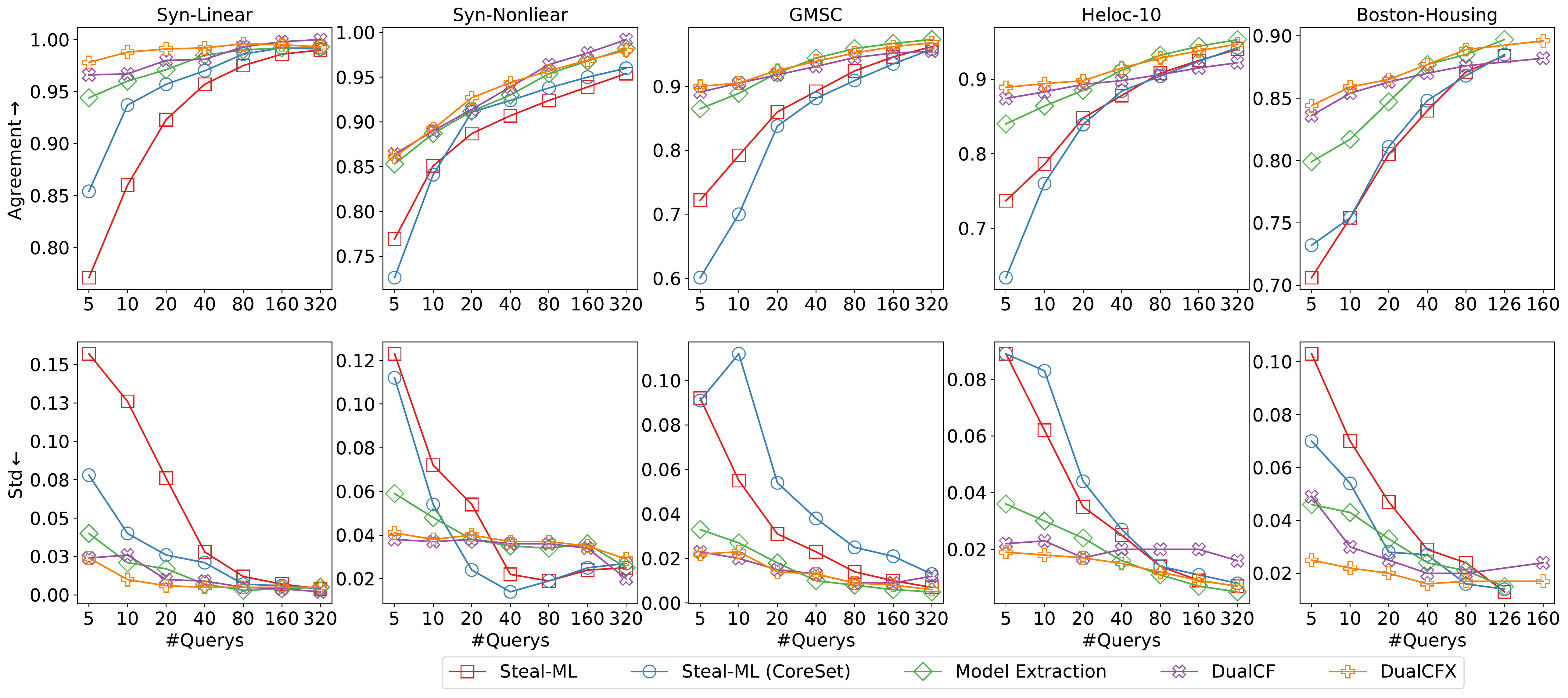}
    \caption{Experiment results on synthetic and real-life datasets. The upper figures report the average agreement while the bottom figures report the standard deviation of agreements of $100$ runs for a fixed size of queries.  $\uparrow$/$\downarrow$ mean the higher/lower results are better.}
    \label{fig:exp_results}
\end{figure*}

We report competitive experimental results on synthetic and real-life datasets in Fig.~\ref{fig:exp_results}. The x-axis represents the size of queries in each run. The y-axis of the upper figures represents the average agreement of $100$ runs while the y-axis of the bottom figures describes the standard deviation of agreements of $100$ runs. From them, we can see (1) Extraction methods on counterfactual explanations achieve better results that those merely on prediction output. This is consistent with our previous analysis where counterfactual explanations leak the boundary information of the cloud model which can help the model extraction; (2) DualCF and DualCFX achieve the best agreement on five datasets when the query size is small; (3) DualCFX will improve the DualCF slightly because it takes the queries into training. More training data will result in better performance. Therefore, we suggest adopting the DualCFX as the default method; (4) The standard deviation of agreements of our proposed methods is smaller than baseline methods in general, which means our methods are less sensitive to the sampling procedure; (5) 
When the query size increases, the gaps between all algorithms decrease because too many predictions from the cloud model are enough to train a good substitute model. 

\subsection{Ablation Studies}

Next, we perform several ablation studies to investigate the effects of some experiment factors on extraction performance. 

\begin{figure*}
    \centering
    \includegraphics[width = \linewidth]{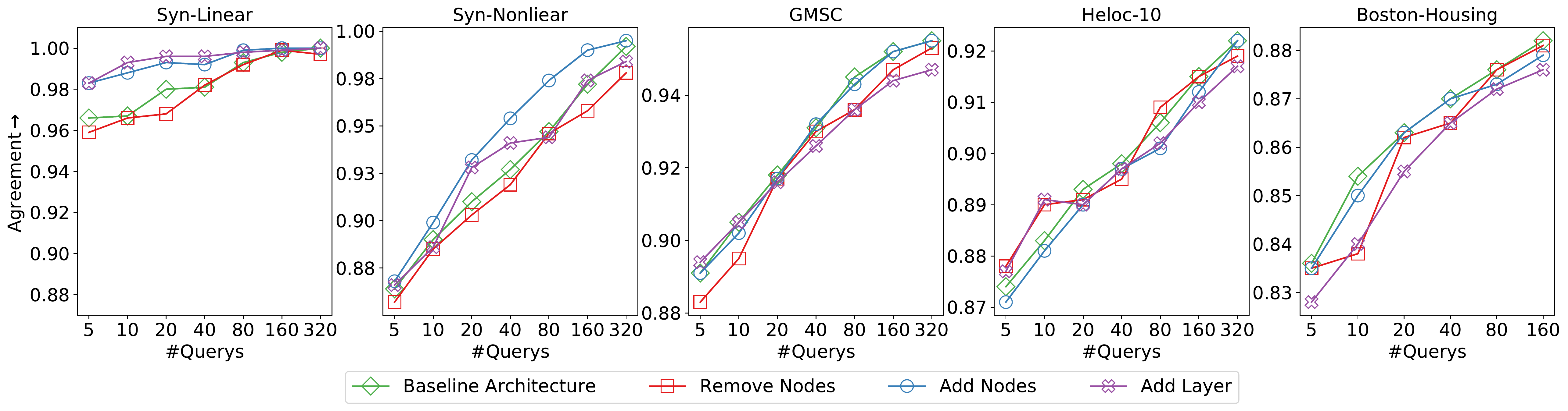}
    \caption{Experiment comparisons on different model capacities in proposed DualCF. ``Remove Nodes'', ``Add Nodes'', and ``Add Layer'' represent three variants of baseline architectures. }
    \label{fig:architecture}
    
    \centering
    \includegraphics[width = \linewidth]{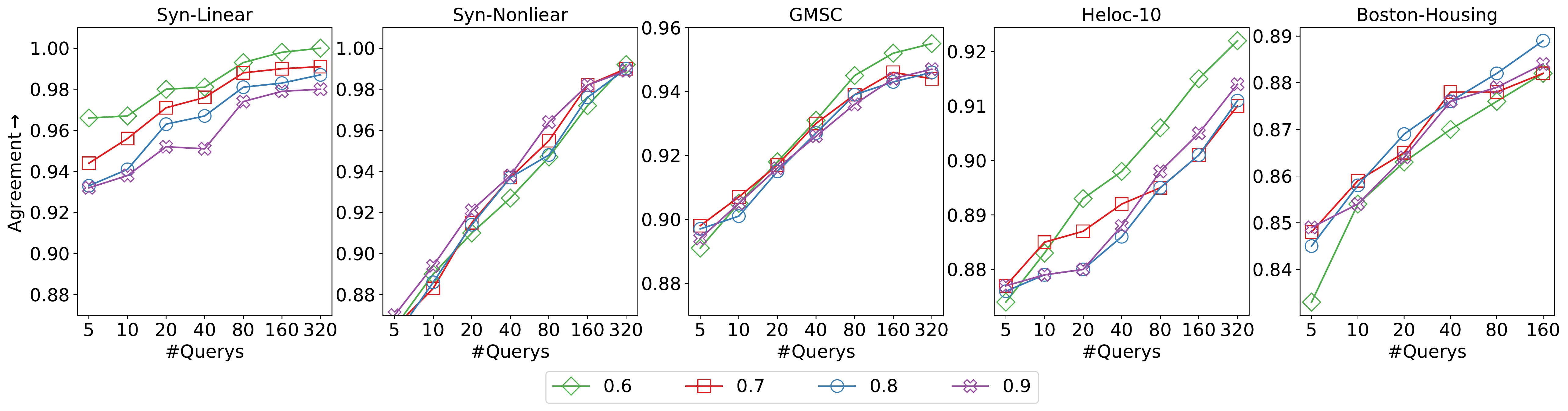}
    \caption{Experimental comparisons on different thresholds in counterfactual explanation generation in our proposed DualCF. We select different thresholds \{0.6, 0.7, 0.8, 0.9\} as the prediction confidence to generate different CFs and CCFs.}
    \label{fig:threshold}
    
    \centering
    \includegraphics[width = \linewidth]{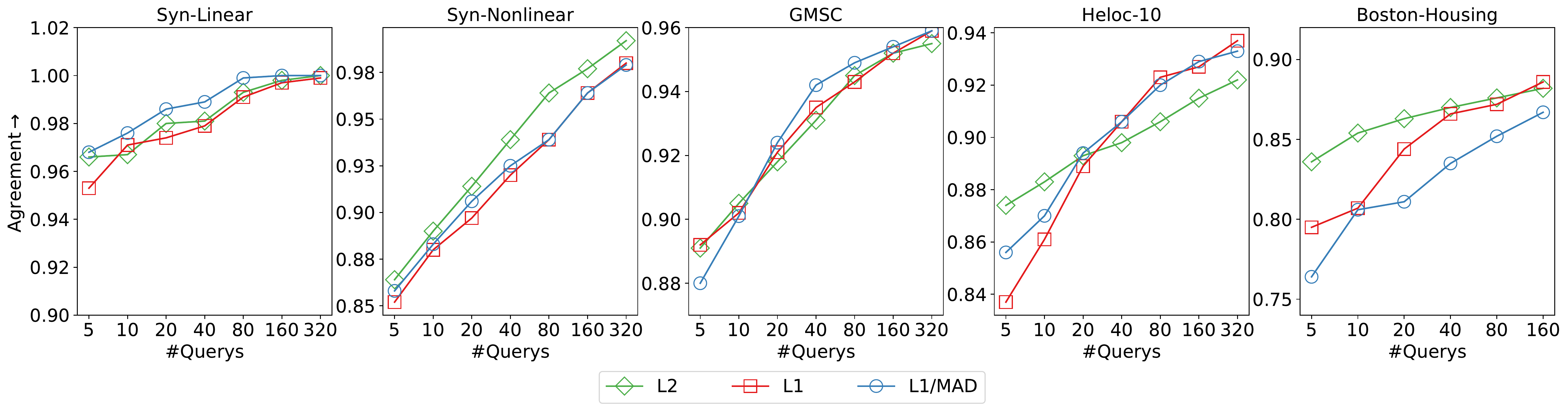}
    \caption{Experimental comparisons on different distance metrics in our proposed DualCF.}
    \label{fig:distance}
    
\end{figure*}


\textbf{Model Capacity.} We first study how capacity of the substitute model affects the extraction results. The baseline substitute model is a MLP classifier, and we consider three variants of the baseline architecture: removing 50\% nodes of last layer, adding 50\% nodes to the last layer, and adding one layer. We detail the baseline architectures and their variants on each dataset in Table 2 in Appendix A.3. We train these models on the same query set and counterfactual explanations, and then report the agreement in Fig. \ref{fig:architecture}. We can see the differences between substitute models with  different capacities are inconsiderable, implying that the model capacity does not play a key role in the extraction attack. A similar conclusion can be found in \cite{papernot2017practical, aivodji2020model}. As the architectures of the cloud model are often unknown to adversaries, selecting the best model is challenging. Adversaries can select a substitute model based on prior knowledge, e.g., the state-of-the-art models for the same or similar task.

\textbf{Threshold.} Counterfactual explanation methods search an explanation above a probability threshold of target class. Here we empirically investigate the influence of different thresholds. We set the threshold from 0.6 to 0.9 by 0.1 and generate different counterfactual explanations for each threshold with the same method. Then, we train a substitute model with DualCF on CFs and CCFs for each threshold respectively. The experiment results are shown in Fig.~\ref{fig:threshold}. The agreement differences between different thresholds are slight on the five datasets. This is consistent to our discussion in Section \ref{subsec:over_confidence}. In the training stage, the cloud model tends to assign high probability to the training points to minimize the training loss. With a higher probability threshold, CF and CCF are still close to the decision boundary of cloud model. 


\textbf{The Distance Metric $d(\cdot)$.} A proper distance metric is essential for finding meaningful proximal instances. Here, we investigate the efficacy of the following three commonly used distance metrics, i.e., 1) $L1$: $\sum_{i=1}^d (\boldsymbol{x}_i - \boldsymbol{c}_i)^2$; 2) $L2$: $\sum_{i=1}^d |\boldsymbol{x}_i - \boldsymbol{c}_i|$; and 3) L1/MAD: $\sum_{i=1}^d \frac{|\boldsymbol{x}_i - \boldsymbol{c}_i|}{MAD_i}$ with $MAD_i$ is the  median absolute deviation of feature $i$ over the training set. 

As shown in Fig.~\ref{fig:distance}, our proposed model is consistently robust to the choice of distance metric on synthetic datasets. On the rest datasets, however, such robustness does not stand when the number of queries is small, i.e., L2 distance metric outperforms $L1$ and $L1/MAD$ when the queries are limited. We conjecture that it is due to the sparsity nature of $L1$-based metrics, where such sparse explanations do not fit the high-dimensional input space.  


\textbf{Imbalanced Dataset.} This experiment was conducted to illustrate that performance of attack methods based on counterfactual explanations are more stable on the imbalanced dataset. In this experiment, we choose an imbalanced GMSC query dataset where the ``Good'' applicants are $5$ times more than ``Bad'' applicants. The evaluation set is kept unchanged. This is the general case in some applications like fraud detection, healthcare. We report  experiment results on imbalanced query set in Fig. \ref{fig:imbalanced}. Compared with the results on balanced datasets in column 3 in Fig.~ \ref{fig:exp_results}, we can see the variance of attack methods based on cloud model's prediction are larger. However,  attack methods for counterfactual explanations are more stable due to the balanced classes during the substitute model training.

    

\begin{figure*}
\centering
\begin{minipage}[t]{.5\linewidth}
  \includegraphics[width=\linewidth]{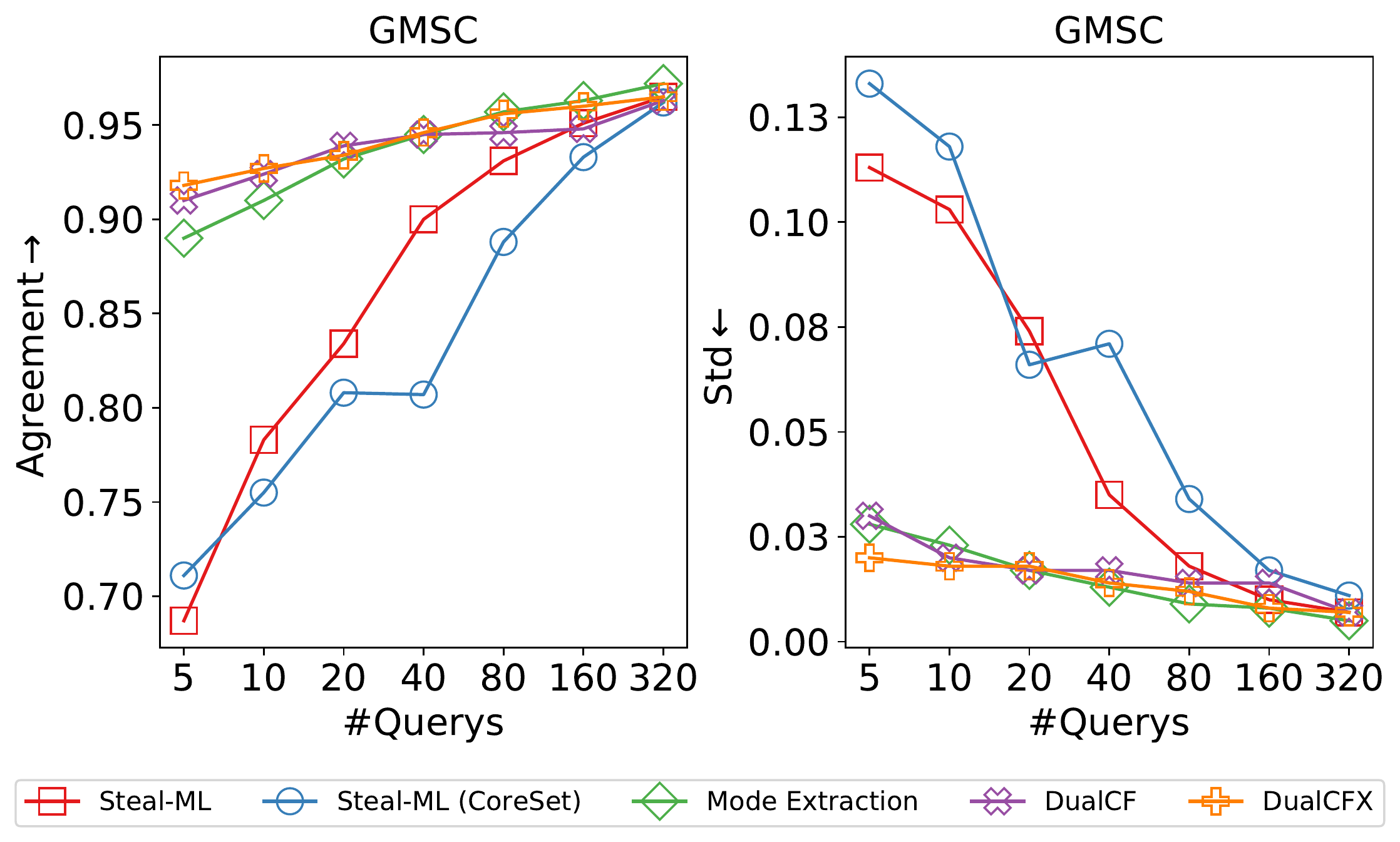}
  \captionof{figure}{Experiment results on imbalanced GMSC query set.}
  \label{fig:imbalanced}
\end{minipage}%
\begin{minipage}[t]{.5\linewidth}
  \includegraphics[width=\linewidth]{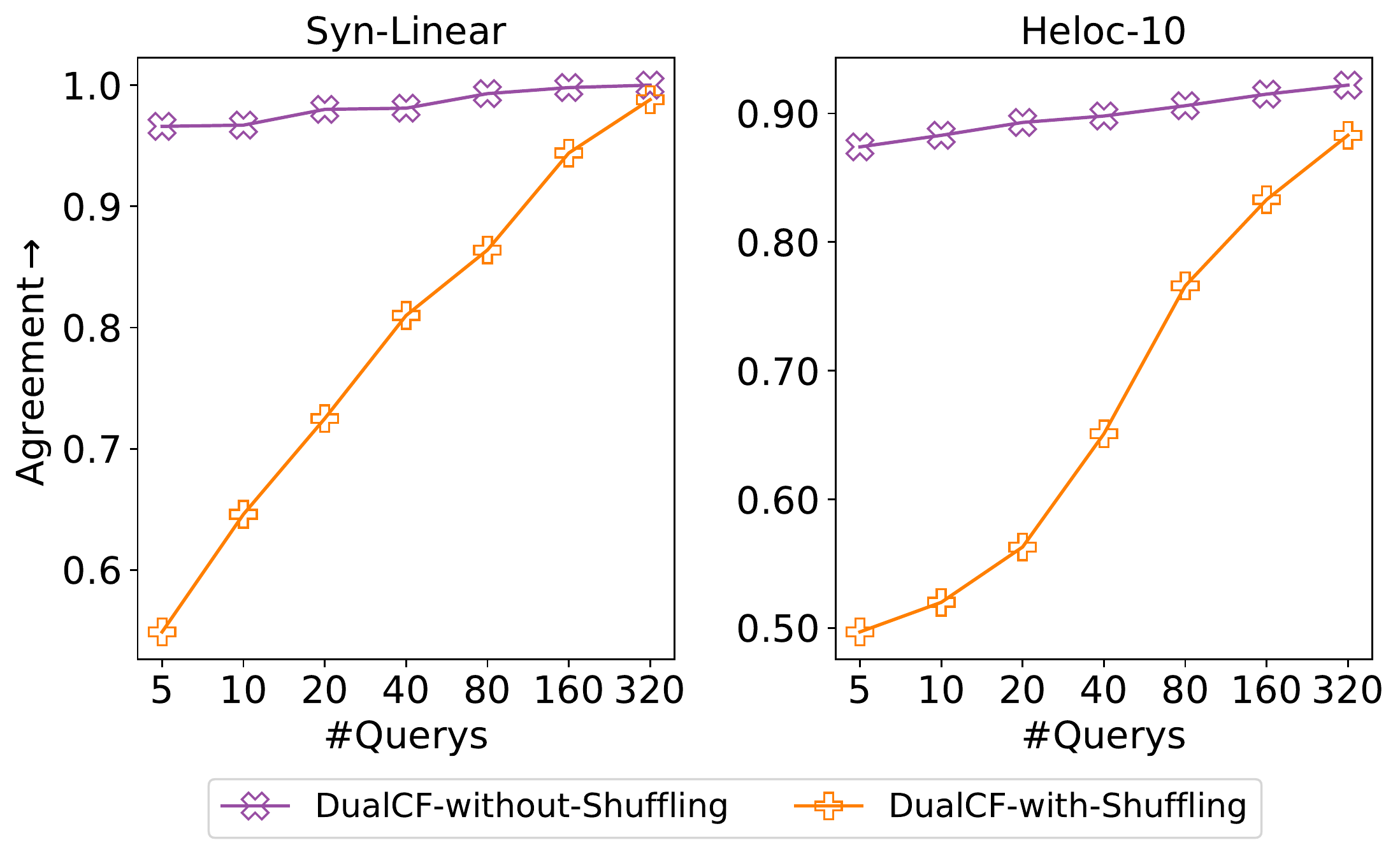}
  \captionof{figure}{Experiment results on with/without shuffling for DualCF.}
  \label{fig:shuffling}
\end{minipage}
\end{figure*}

\textbf{Shuffling CFs and CCFs.} We claim that CF and its CCF should be in the same batch during training. Fig.~\ref{fig:shuffling} reports the experiment difference between with/without shuffling CFs and CCFs. Experiment results reveal that training with pairs of CFs and CCFs in the same batch achieves better results. The shuffle operation will destroy the propose of our method that we should train on pairs of instances with similar distance to distance boundary of cloud model and with different predictions.

\section{Conclusion}

Counterfactual explanation method searches the minimum perturbations for an input to achieve a particular different prediction, which has broad applications. However, the boundary information leaked from counterfactual explanations is prone to model extraction attacks. In this paper, we propose a simple yet efficient method DualCF that mitigates the decision boundary shift issue in existing methods. Extensive experiments demonstrate that our method can achieve a high-fidelity model with much fewer queries. Our work raises the awareness of the privacy and security issues of counterfactual explanations, and further motivates the countermeasures (e.g., monitoring malicious queries, adding superfluous features on model training and counterfactual explanations for misleading adversaries, restrict one-way generation of counterfactual explanations.) to protect the cloud model. In the future, we will work on more general attacks and defenses on multi-class classification models and regression models, and explore the influence of different practical constraints on the security of counterfactual explanations. 

\begin{acks}
This research is supported by the National Research Foundation, Prime Minister’s Office, Singapore under its NRF Investigatorship Programme (NRFI Award No. NRF-NRFI05-2019-0002). Any opinions, findings and conclusions or recommendations expressed in this material are those of the author(s) and do not reflect the views of National Research Foundation, Singapore.
This research is supported, in part, by Alibaba Group through Alibaba Innovative Research (AIR) Program and Alibaba-NTU Singapore Joint Research Institute (JRI), Nanyang Technological University, Singapore.  H.Qian thanks the support from the Wallenberg-NTU Presidential Postdoctoral Fellowship.
\end{acks}

\bibliographystyle{ACM-Reference-Format}
\bibliography{sample-base}

\clearpage
\begin{center}\large\bfseries
Appendix for ``DualCF: Efficient Model Extraction Attack from Counterfactual Explanations''
\end{center}
\appendix

\section{Experiments details}

\subsection{Counterfactual Explanations Visualization on Synthetic Datasets}
\label{sec:visualization}
\begin{figure*}[htbp]
    \centering
    \includegraphics[width = 0.45\linewidth]{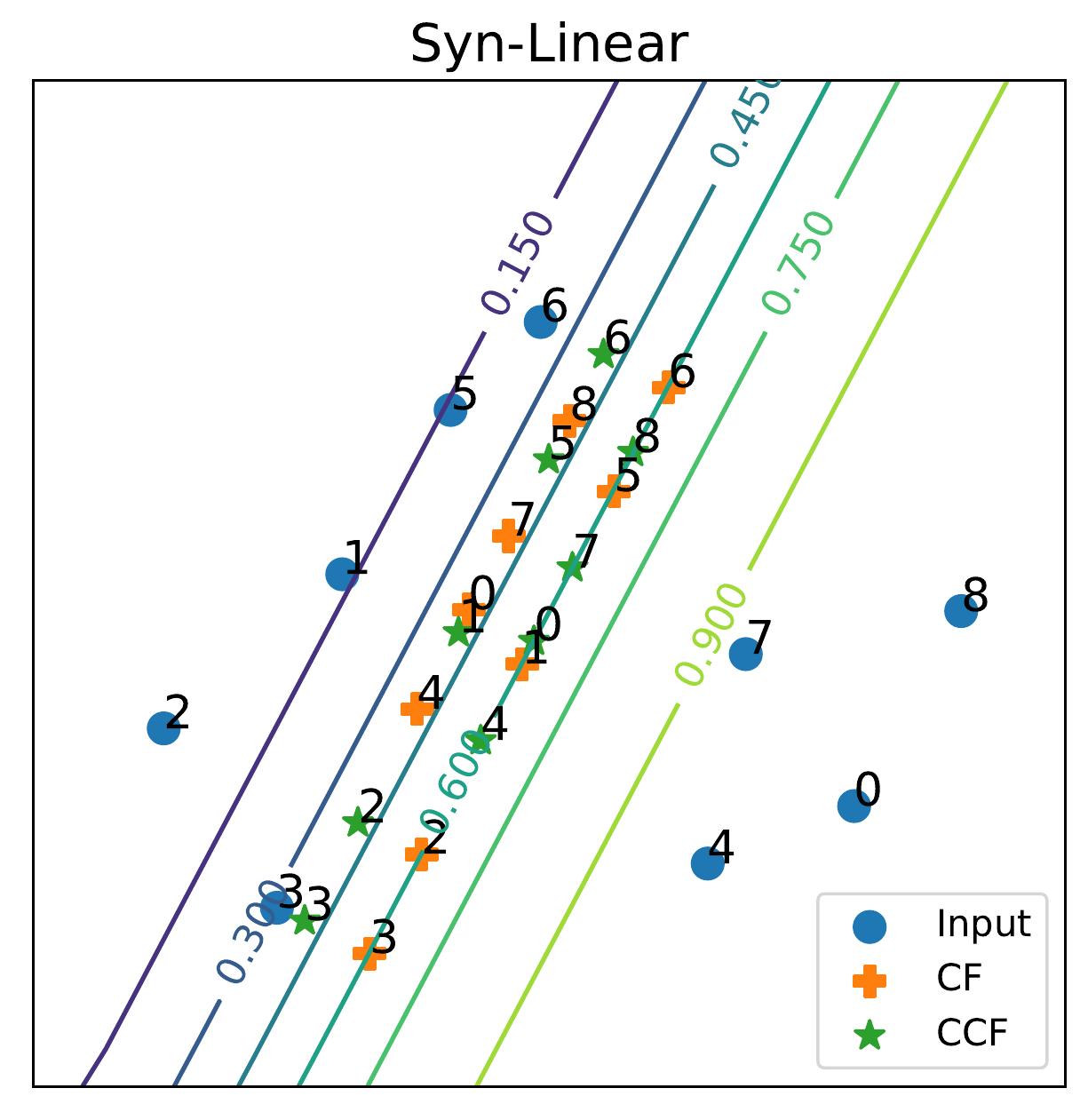}
    \includegraphics[width = 0.45\linewidth]{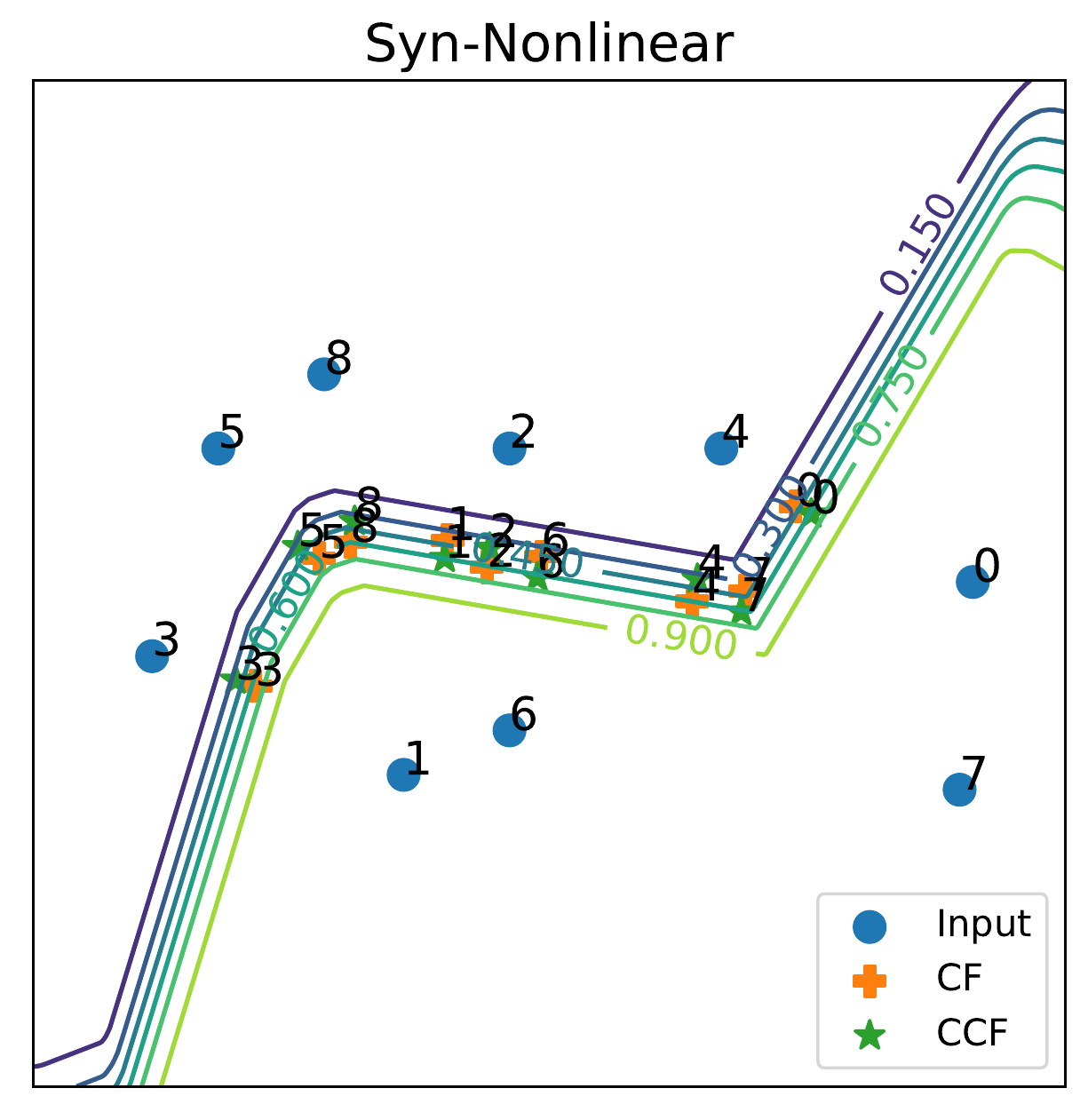}
    \caption{Counterfactual explanations visualization on Syn-Linear and Syn-Nonlinear datasets.}
    \label{fig:cf_vis}
\end{figure*}

We visualize counterfactual explanations of sampled queries on two synthetic datasets in the Fig.  ~\ref{fig:cf_vis}. From the figure, we can see all counterfactual explanations locate in the close-to-boundary regions. The CCFs and queries have the same prediction from the cloud model. CCFs and CFs have similar distance to the decision boundary of cloud model and have opposite predictions. 

\subsection{Training Details of Substitute Model}
\label{sec:details}

We report the training details of substitute model on five datasets in Table ~\ref{tab:details}. The substitute models are MLP. We adopt the Adam optimizer to minimize the binary cross entropy loss for substitute model. The learning rate and training epoches are tuned by the curve of training loss. 

\begin{table*}[h]
\caption{The model architectures, hyper-parameters used for training substitute models on synthetic and real-life datasets.}
\label{tab:details}
\resizebox{\textwidth}{!}{%
\begin{tabular}{|l|l|l|l|l|l|}
\hline
Datasets & Syn-Linear & Syn-Nonliear & GMSC & Heloc-10 & Boston-Housing \\ \hline
Architecture & \begin{tabular}[c]{@{}l@{}}Linear(2, 10), ReLU, \\ Linear(10, 1), Sigmoid\end{tabular} & \begin{tabular}[c]{@{}l@{}}Linear(2, 20), ReLU,\\ Linear(20, 10), ReLU,\\ Linear(10, 1), Sigmoid\end{tabular} & \begin{tabular}[c]{@{}l@{}}Linear(12, 20), ReLU,\\ Linear(20, 10), ReLU,\\ Linear(10, 1), Sigmoid\end{tabular} & \begin{tabular}[c]{@{}l@{}}Linear(10, 20), ReLU,\\ Linear(20, 10), ReLU,\\ Linear(10, 1), Sigmoid\end{tabular} & \begin{tabular}[c]{@{}l@{}}Linear(10, 20), ReLU,\\ Linear(20, 10), ReLU,\\ Linear(10, 1), Sigmoid\end{tabular} \\ \hline
Optimizer & Adam & Adam & Adam & Adam & Adam \\ \hline
Loss function & \begin{tabular}[c]{@{}l@{}}Binary cross entropy \\ (BCE)\end{tabular} & BCE & BCE & BCE & BCE \\ \hline
Learning Rate & 0.005 & 0.005 & 0.01 & 0.01 & 0.005 \\ \hline
Batch Size & 32 & 32 & 32 & 32 & 32 \\ \hline
Epoch & 200 & 500 & 200 & 200 & 200 \\ \hline
\end{tabular}%
}
\end{table*}

\subsection{Details of Model Capacity Experiments}
\label{sec:capacities}

Here, we list the models used in model capacity studies in Table~\ref{tab:capacities}. We train the baseline architectures and their three variants on the same experiment settings. ``Remove Nodes'' and ``Add Nodes'' represent that we remove $50\%$ nodes of the last layer and add $50\%$ nodes on the last layer. ``Add Layer'' means we add one layer. 

\begin{table*}[htbp]
\caption{The models used in model capacity studies.}
\label{tab:capacities}
\resizebox{\textwidth}{!}{%
\begin{tabular}{|l|l|l|l|l|l|}
\hline
Datasets & Syn-Linear & Syn-Nonliear & GMSC & Heloc-10 & Boston-Housing \\ \hline
\begin{tabular}[c]{@{}l@{}}Base\\ Architectures\end{tabular} & \begin{tabular}[c]{@{}l@{}}Linear(2, 10), ReLU, \\ Linear(10, 1), Sigmoid\end{tabular} & \begin{tabular}[c]{@{}l@{}}Linear(2, 20), ReLU,\\ Linear(20, 10), ReLU,\\ Linear(10, 1), Sigmoid\end{tabular} & \begin{tabular}[c]{@{}l@{}}Linear(12, 20), ReLU,\\ Linear(20, 10), ReLU,\\ Linear(10, 1), Sigmoid\end{tabular} & \begin{tabular}[c]{@{}l@{}}Linear(10, 20), ReLU,\\ Linear(20, 10), ReLU,\\ Linear(10, 1), Sigmoid\end{tabular} & \begin{tabular}[c]{@{}l@{}}Linear(10, 20), ReLU,\\ Linear(20, 10), ReLU,\\ Linear(10, 1), Sigmoid\end{tabular} \\ \hline
Remove Nodes & \begin{tabular}[c]{@{}l@{}}Linear(2, 5), ReLU, \\ Linear(5, 1), Sigmoid\end{tabular} & \begin{tabular}[c]{@{}l@{}}Linear(2, 20), ReLU,\\ Linear(20, 5), ReLU,\\ Linear(5, 1), Sigmoid\end{tabular} & \begin{tabular}[c]{@{}l@{}}Linear(12, 20), ReLU,\\ Linear(20, 5), ReLU,\\ Linear(5, 1), Sigmoid\end{tabular} & \begin{tabular}[c]{@{}l@{}}Linear(10, 20), ReLU,\\ Linear(20, 5), ReLU,\\ Linear(5, 1), Sigmoid\end{tabular} & \begin{tabular}[c]{@{}l@{}}Linear(10, 20), ReLU,\\ Linear(20, 5), ReLU,\\ Linear(5, 1), Sigmoid\end{tabular} \\ \hline
Add Nodes & \begin{tabular}[c]{@{}l@{}}Linear(2, 15), ReLU, \\ Linear(15, 1), Sigmoid\end{tabular} & \begin{tabular}[c]{@{}l@{}}Linear(2, 20), ReLU,\\ Linear(20, 15), ReLU,\\ Linear(15, 1), Sigmoid\end{tabular} & \begin{tabular}[c]{@{}l@{}}Linear(12, 20), ReLU,\\ Linear(20, 15), ReLU,\\ Linear(15, 1), Sigmoid\end{tabular} & \begin{tabular}[c]{@{}l@{}}Linear(10, 20), ReLU,\\ Linear(20, 15), ReLU,\\ Linear(15, 1), Sigmoid\end{tabular} & \begin{tabular}[c]{@{}l@{}}Linear(10, 20), ReLU,\\ Linear(20, 15), ReLU,\\ Linear(15, 1), Sigmoid\end{tabular} \\ \hline
Add Layer & \begin{tabular}[c]{@{}l@{}}Linear(2, 10), ReLU, \\ Linear(10, 10), ReLU,\\ Linear(10, 1), Sigmoid\end{tabular} & \begin{tabular}[c]{@{}l@{}}Linear(2, 20), ReLU,\\ Linear(20, 10), ReLU,\\ Linear(10, 10), ReLU,\\ Linear(10, 1), Sigmoid\end{tabular} & \begin{tabular}[c]{@{}l@{}}Linear(12, 20), ReLU,\\ Linear(20, 10), ReLU,\\ Linear(10, 10), ReLU,\\ Linear(10, 1), Sigmoid\end{tabular} & \begin{tabular}[c]{@{}l@{}}Linear(10, 20), ReLU,\\ Linear(20, 10), ReLU,\\ Linear(10, 10), ReLU,\\ Linear(10, 1), Sigmoid\end{tabular} & \begin{tabular}[c]{@{}l@{}}Linear(10, 20), ReLU,\\ Linear(20, 10), ReLU,\\ Linear(10, 10), ReLU,\\ Linear(10, 1), Sigmoid\end{tabular} \\ \hline
\end{tabular}%
}
\end{table*}

\end{document}